\documentclass[12pt,draftcls,journal,onecolumn]{IEEEtran}

\usepackage{graphicx,multirow}
\usepackage{epstopdf}
\usepackage[latin1]{inputenc}
\usepackage{amsmath,amssymb,amscd,latexsym,dsfont,amsfonts}
\usepackage{tabularx}
\usepackage{graphicx}
\usepackage{algpseudocode}
\usepackage{algorithm}
\usepackage{algorithmicx}
\usepackage{float}
\usepackage{psfrag}
\usepackage{comment,psfrag,subfigure,enumerate,cite}
\usepackage{hyperref}
\usepackage[all]{xy}
\usepackage{multirow}
\newtheorem{theorem}{Theorem}
\newtheorem{lemma}{Lemma}
\newtheorem{corollary}{Corollary}
\newenvironment{remark}[1][Remark]{\begin{trivlist}
\item[\hskip \labelsep {\bfseries #1}]}{\end{trivlist}}

\begin{document}

\title{On Optimum Causal Cognitive Spectrum Reutilization Strategy}
\author{Kasra Haghighi, Erik G. Str\"{o}m and Erik Agrell
\thanks{The authors are with Chalmers Univ. of Technology, Sweden (email: kasra.haghighi@chalmers.se). Research supported by the High Speed Wireless Communication Center, Lund and the Swedish Foundation for Strategic Research. This work has been submitted to the IEEE for possible publication. Copyright may be transferred without notice, after which this version may no longer be accessible. }
}
\maketitle
\begin{abstract}
In this paper we study opportunistic transmission strategies for cognitive radios (CR) in which causal noisy observation from a primary user(s) (PU) state is available. PU is assumed to be operating in a slotted manner, according to a two-state Markov model. The objective is to maximize utilization ratio (UR), i.e., relative number of the PU-idle slots that are used by CR, subject to interference ratio (IR), i.e., relative number of the PU-active slots that are used by CR, below a certain level. We introduce an a-posteriori LLR-based cognitive transmission strategy and show that this strategy is optimum in the sense of maximizing UR given a certain maximum allowed IR. Two methods for calculating threshold for this strategy in practical situations are presented. One of them performs well in higher SNRs but might have too large IR at low SNRs and low PU activity levels, and the other is proven to never violate the allowed IR at the price of a reduced UR. In addition, an upper-bound for the UR of any CR strategy operating in the presence of Markovian PU is presented. Simulation results have shown a more than 116\% improvement in UR at SNR of $-3$dB and IR level of 10\% with PU state estimation. Thus, this opportunistic CR mechanism possesses a high potential in practical scenarios in which there exists no information about true states of PU.
\end{abstract}

\section{Introduction}
\IEEEPARstart
The limited availability of radio spectrum, together with the ever increasing demands for data rates, has created a big challenge for spectrum regulators, manufacturers and operators as they need to meet the demand. Modulation and coding are approaching the Shannon limits, which makes the higher spectral efficiencies theoretically impossible\cite{dohler2011}. On the other hand, the hardware impairments including but not limited to power amplifiers nonlinearities, analog to digital conversion issues and phase noise limit the efficient use of frequency bands. Although the usable spectrum is limited, Federal Communications Commission (FCC) studies have shown that the spectrum is severely underutilized\cite{FCC_spectrum}. More specifically, studies have shown that the utilization of the spectrum in different geographical areas varies significantly. For example, fading in primary wireless channels creates spatial spectrum holes which can be exploited by secondary users\cite{Molisch2009,Haykin2005}. The introduction of software defined radio is an enabling technology for the dynamic spectrum access\cite{Haykin2005,Mitola1993}, which motivates the reuse of the unhindered spectrum. The concept of cognitive radio (CR) as defined first by J. Mitola\cite{Mitola1993} entails that the communication devices adapt themselves to the spectrum\cite{Haykin2005}.\\
\indent In the context of CR, spectrum sensing plays a crucial role for the cognition phase. Since the spectrum sensing is affected by the type of signal detectors e.g., energy detectors, match filter detectors, cyclostationary feature detectors, wavelet feature detectors, etc., the measure of the performance of a CR is normally based on the performance of its spectrum sensor\cite{Yecek2009}. Usually, detectors and spectrum sensing algorithms are characterized by their probabilities of mis-detection and false-alarm\cite{Poor1994}\cite{Yecek2009}. However, the obvious choice of using these probabilities might not the best choice to serve the purpose of cognition and adaptation of CRs. These two probabilities  carry information only about a detector and not the interaction between the primary user of the band and CR transmission strategies. Some researchers approached performance evaluation of CRs from the capacity point of view\cite{Haddad2007}, which is valid with a sophisticated channel code and a large block length (delay). Thus, a need for proper measures for evaluating the performance of cognitive radios (networks) emerges.\\
\indent In the traditional implementations of CR, in which only the current sensed received signal is considered for the transmission decision in the succeeding time slots, the important fact that the PU traffic might be according to a certain model is ignored. CR also expects that its observation resembles the true transmission state of PU, and PU will not change its state in the period of CR transmission. Clearly, since this CR does not incorporate the PU transmission model in its decision, the performance of CR will improve if the CR decision algorithm includes such a model.
 This will require a beyond-PHY or cross-layer design. Thus, integrating the PU model into CR transmission strategy will enable CR to have credible prediction of PU states.\\
\indent In information theory literature, normally it is assumed that the CR(s) have non-causal information about PU(s) activities through a process called genie\cite{Devroye2006}. However, in practical applications this assumption does not hold. Many researchers use only the current state of PU for transmission in the slot.\\
\indent In addition, CRs suffer from other problems. The capabilities of CRs utilizing energy detection spectrum sensing is limited by the SNR wall\cite{tandra2008}. This is due to the low received power of the PU signal at the CR receiver and uncertainties in signals, noise, and channels. This effect is more visible\cite{haghighi2011,haghighi20113} in wideband spectrum sensing in particular. This can ultimately result in large sensing delays. Nevertheless, spectrum opportunities appear and disappear quickly, and they depend on the occupancies in different bands. A real cognitive radio, which, according to the cognitive cycle\cite{Mitola1993,Akyildiz2006} should adapt itself to the dynamics of the spectrum, needs to be agile to react to the changes in the spectrum\cite{Chou2007} as fast as possible. On the other hand, in some cases such as energy detectors, agility compromises the accuracy of sensing the spectrum, which ultimately jeopardizes not only interference level made for PU but also reduces the spectrum reuse. Thus, a CR which can optimally incorporate all previous observations and thus decides for transmission within a short time is appealing. Sequential spectrum sensing has been proven to be on average faster than traditional energy detection\cite{Poor1994,poor2009,haghighi2009,haghighi2010}. However, since detection time varies in sequential detection, it is not a good candidate for slotted CR strategy.\\
\indent In this manuscript we deploy a hidden Markov model (HMM) to form a framework for modeling the behavior of CRs in the presence of PUs and all the uncertainties. Additionally, a benchmark for evaluation of CR performance is introduced. Then, using this foundation and these measures, a new CR transmission strategy is designed and implemented. This new design ensures that the vacant spectrum is optimally used conditioned on the level of interference for the PU, because of all uncertainties in the model, is not exceeding a certain level.\\
\indent HMMs are long in use for modelling different phenomena ranging from speech signals\cite{Rabinner1989} to the complex behavior of computer networks. In the context of cognitive radio, many researchers model the spectrum white space with Markov models and spectrum sensing using HMMs\cite{ChenZ2009,ghosh2010,Coulson2009,Akbar2007,XiaoLZ09,Li2009,zhang2011}. In our paper, HMMs are used not only for spectrum sensing but also as a tool for CR transmission strategy making. The closest published approach to our method is presented in~\cite{Zhao2007,Choi2011}, which employs a partially observed Markov decision process. They used this process for optimal policy making for multiple channel sensing and access. The approach is similar to ours due to the Markovian assumption for the PU transmission model and in the presence of sensing errors. However, the sensing model, performance metric, and constraints are different from ours.\\
\indent To summarize the contributions of this paper following items can be enlisted
\begin{itemize}
\item A new performance measure for characterizing CR performance is introduced
\item A novel APP-LLR based opportunistic spectrum reutilization strategy is proposed
\item Optimality of this new strategy is proved
\item Two practical methods for calculating threshold for APP-LLR based strategy are introduced, one is suitable for high-SNR regimes and presents close to optimum URs but IR may be too high at low SNR. The other never violates the allowed IR level, but with a reduced UR,
\item An upper bound on the UR for any CR transmission strategy is established.
\end{itemize}

\section{System model}
\indent This section presents the model which accounts for the PU signal and noise. First, a more general perspective is demonstrated and then a simplified version will be used.
\subsection{Complete PU transmission model}
\indent A cognitive radio system is designed to utilize the spectrum vacancies. To take advantage of time-frequency slots which are not used by the PU, the CR must be aware of the PU activities. In this paper, it is assumed that the CR has a full buffer to reuse the spectrum whenever it is available.\\
\indent CR will receive the PU signal which is attenuated by the PU-CR channel. If there exists more than one PU in the vicinity of the CR, the aggregated signal will be received by the CR antenna. It is possible to assume that the PUs operating in the same frequency band and co-located, belong to the same network and thus from the CR point of view can be modelled as a single entity. Since protection of each one of the PUs is as important as the others, a network of PUs for CR can be represented by a single but more active PU, although this would yield a suboptimal CR performance compared with a multi-PU model.\\
\indent Another factor in modelling the PU-CR interaction is the channel in between. Wireless channels are normally considered as random fading processes such as Rayleigh, Rician, Nakagami, etc\cite{Digham2007,Atapattu2009}. For simplicity it can be assumed that the fading gain is constant and known during the operation of this CR. Another approach to model the fading process is to include the fading in the PU transmission model. Thus, whenever channel is in a deep fade, it is assumed that there is no PU transmission, no matter what the real state of the PU is. And in case of no deep fade, the standard PU transmission model will be deployed. With this brief introduction, a simple two-state Markov model can approximate a wide range of PU transmissions, PU network activities and even fading channels.
In the next section, the simplified two-state Markov model will be presented as the PU transmission model.
\subsection{Simplified PU transmission model}\label{sys_model}
\indent Now, the PU transmissions are assumed to be slotted, since in most of today's digital communication systems transmissions are confined within a packet, frame or generally some block structure of some minimum length $T_{\text{F}}$. However, the CR is expecting PU activities and vacancies in much smaller slots of length $T\ll T_{\text{F}}$. Smaller slot size improves the agility of CR to adapt its transmission to the PU activity. For the sake of simplicity, we will assume that the CR slots are synchronized to the PU slots. However, because of the small CR slot length in comparison to the PU slot length, mismatches in synchronization will not cause major performance degradation. The existence of a PU transmission in slot $k$ i.e., during time $t \in [kT, (k+1)T)$, is denoted by the hypothesis $H_1\triangleq\{q_k=1\}$ and its absence is denoted by $H_0\triangleq\{q_k=0\}$. A simple model which represents the PU transmission is the two-state on-off Markov process depicted in Fig. \ref{ghmm}, where the Markov chain is represented by the transition probabilities $a_{i,j}=\Pr\{q_{k+1}=j|q_k=i\}>0$ for $i,j\in\{0,1\}$ and $q_k$ stands for the PU state at time slot $k$. The transition matrix is
 \begin{align}
\mathbf{A}\triangleq\left[
\begin{array}{c c}
a_{00} & a_{01} \\
a_{10} & a_{11}
\end{array}\right]\label{Amat},~a_{00}+a_{01}=a_{10}+a_{11}=1.
 \end{align}
\indent The initial distribution of the states is assumed to be in a steady state \cite{Rabinner1989} and defined as
 \begin{align}
\boldsymbol{\pi}&\triangleq\left[
\begin{array}{c c}
\pi_0 & \pi_1
\end{array}\right]\triangleq\left[
\begin{array}{c c}
\Pr\{q_k=0\} & \Pr\{q_k=1\}
\end{array}\right]=\left[
\begin{array}{c c}
\frac{a_{10}}{a_{01}+a_{10}} & \frac{a_{01}}{a_{01}+a_{10}}
\end{array}\right],~k=0,1,2,\cdots\label{pi}
\end{align}
It is assumed that the PU activities happen with a period $T_{\text{F}}$ larger than the period CR Markov chain operating on $T$. Thus, the chance of staying in one state or another is much higher than the chance of transition between states. This allows us to assume that $a_{01}+a_{10}< 1$, which turns out to be useful in Section \ref{nopustate}.
\subsection{Signal and noise model}\label{signoise}
\begin{figure}[!t]
  \centering
  \psfrag{p}[][][2]{$a_{01}$}
  \psfrag{q}[][][2]{$a_{10}$}
  \psfrag{p2}[][][2]{$a_{00}$}
  \psfrag{q2}[][][2]{$a_{11}$}
  \psfrag{h0}[][][2]{$q_k=0$}
  \psfrag{h1}[][][2]{$q_k=1$}
  \scalebox{.4}{\includegraphics{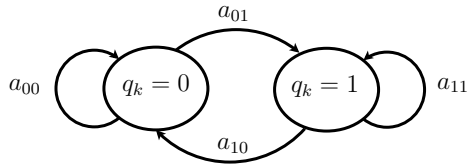}}
  \caption{PU transmission model}
 \label{ghmm}
\end{figure}
\indent The receiver front end is an energy detector whose output is $y_k\triangleq\sum_{i=0}^{K-1}\left|r\left(kT+iT_s\right)\right|^2$, where $r(\cdot)$ is the complex envelope of the received signal low-pass filtered to the PU signal bandwidth $W$, $T$ is the period in which energy is collected, $T_s$ is the sampling time, and $K$ is the total number of samples in each period. We assume that the received PU signal can be modelled as a Gaussian random process. The Gaussian PU signal model is common in literature\cite{Lai2008}\cite{Haykin2005}, and is reasonable for many combinations of PU signal formats and channels (fading as well as nonfading).
If we select $T_s$ such that $T_s \gg 1/W$, then the samples $r(iT_s)$ are approximately statistically independent. We note that $K$ is constrained as $K \leq T/T_s$.\\
\indent Since noise and channel uncertainty exists in the CR observation of the PU signal, the true PU state from Fig.~\ref{ghmm} is not observable. Depending on the state of the PU a continuous energy level which consists of noise only or signal plus noise is observed. This model corresponds to a continuous-output HMM depicted in Fig. \ref{xhmm_fix}.

\subsubsection{Noise only}
\indent In state $H_{0}$, the noise $n(iT_s)\sim \mathcal{CN}(0,\sigma^2_0)$ is a zero-mean complex circular Gaussian ($\mathcal{CN}$ stands for complex circular Gaussian) sample with variance $\sigma^2_0$, and the received signal will be $r(iT_s)=n(iT_s)$. Thus, $y_k$ is chi-square distributed with $2K$ degrees of freedom and Gaussian variance $\sigma^2_0/2$.
\subsubsection{Signal plus noise}
\indent In state $H_{1}$, the noise is a zero-mean complex circular Gaussian sample with variance $\sigma^2_0 $, the signal is also zero-mean complex circular Gaussian with variance $\sigma^2_s$, and $r(iT_s)=s(iT_s)+n(iT_s)$, $r(iT_s) \sim \mathcal{CN}(0,\sigma^2_1)$, where $\sigma^2_1=\sigma^2_s+\sigma ^2_0$. Thus, $y_k$ is chi-square distributed with $2K$ degrees of freedom and Gaussian variance $\sigma^2_1/2$.
\begin{figure}[t]
  \centering
  \psfrag{p}[][][2]{$a_{01}$}
  \psfrag{q}[][][2]{$a_{10}$}
  \psfrag{p2}[][][2]{$a_{00}$}
  \psfrag{q2}[][][2]{$a_{11}$}
  \psfrag{v0}[][][2]{$\times$}
  \psfrag{v1}[][][2]{$\times$}
  \psfrag{nn}[][][1.7]{$\mathcal{CN}(0,\sigma^2_0)$}
  \psfrag{sn}[][][1.7]{$\mathcal{CN}(0,\sigma^2_1)$}
  \psfrag{h0}[][][2]{$q_k=0$}
  \psfrag{h1}[][][2]{$q_k=1$}
  \psfrag{r}[][][2]{$r(t)$}
  \psfrag{x}[][][2]{$y_{k}$}
  \psfrag{ED}[][][2]{$\sum|r(.)|^2$}
  \scalebox{.4}{\includegraphics{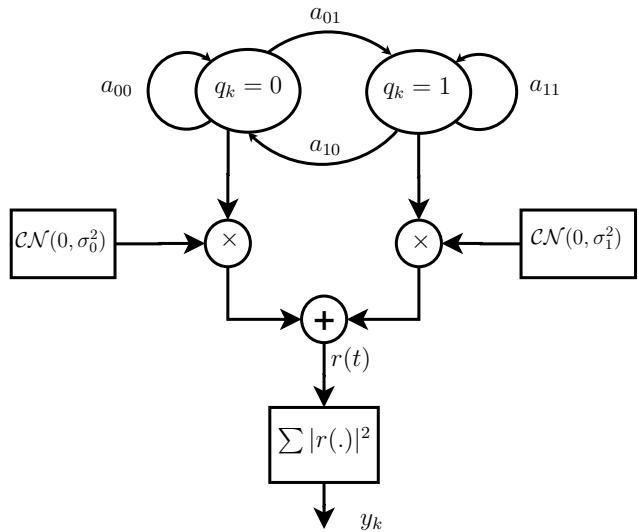}}
  \caption{Continuous-output HMM of received signal at CR}
 \label{xhmm_fix}
\end{figure}

\section{Problem statement and performance metrics}\label{SUandIR}
\indent Cognitive radios exploit channel availability information from spectrum sensing and decide whether to transmit or not. In this paper we assume that the CR has a full buffer to transmit. Thus, it would like to take advantage of any spectral opportunities and transmit whenever possible. However, due to channel and noise uncertainties it will create unintentional interference for PU.  Our goal is to design the best CR transmission strategy denoted by $u_{k+1}$, where $u_{k+1}=0$ and $u_{k+1}=1$ represent no transmission and transmission, respectively in slot $k+1$ using the observations until time $k$, $\mathbf{y}_k\triangleq\left[y_1,y_2,\ldots,y_k\right]^T$. This strategy is supposed to not interfere with the PU more than specific limit.\\

\indent The performance of a CR is usually assessed based on its spectrum-sensing algorithm. Spectrum sensing is judged based on its probability of false-alarm $P_{\text{FA}}$ and probability of mis-detection $P_{\text{M}}$, which are normally presented in receiver operating characteristic plots. However, the ultimate goal of CRs is to reutilize the idle spectrum slots while keeping the level of interference for PUs below a certain level. The two aforementioned measures are not taking the PU behavior into account. Moreover, utilization and interference are defined by the presence or absence of PU transmission. Therefore, it is advantageous to define new criteria which consider the full picture including PUs, CRs, and even the channel.
\subsection{Definitions}
\indent Interference will happen whenever the CR transmits at the same time as the PU. Thus, the interference ratio (IR) $\rho$ is defined as\cite{Haghighi20112}
\begin{align}
\rho\triangleq\Pr\{u_{k+1}=1|q_{k+1}=1\}.\label{IR1}
\end{align}
Utilization of the spectrum occurs whenever the CR transmits in a vacant time--frequency slot. Thus, we define the spectral utilization ratio (UR) as
\begin{align}
\eta\triangleq\Pr\{u_{k+1}=1|q_{k+1}=0\}.\label{SU1}
\end{align}
\indent The intention of any CR is to design a strategy that keeps $\rho$ below a specified level, say $\rho_{\text{max}}$, and then maximizes the utilization ratio $\eta$. Hence, we call a transmission scheme that maximizes $\eta$ while $\rho\leq\rho_{\text{max}}$ an optimal transmission scheme for a given $a_{01}$ and $a_{10}$. The relation of UR and IR to the transmission rate and the probability of error of CR appeared in\cite{Haghighi20112}. The following theorem states that the UR and IR depends on the PU $\mathbf{A}$,  $P_{0}\triangleq\Pr\{u_{k+1}=0|q_k=0\}$ and $P_{1}\triangleq\Pr\{u_{k+1}=1|q_k=1\}$.
\begin{theorem}\label{theory1}
Assume that PU follows the Markov model presented in Fig. \ref{ghmm}. For any CR, the UR and the IR are given by
\begin{align}
\eta&=a_{01}P_{1}+a_{00}(1-P_{0}),\label{SU3}\\
\rho&=a_{11}P_{1}+a_{10}(1-P_{0}).\label{IR3}
\end{align}
\end{theorem}
\begin{proof}
Proof for a similar theorem is presented in\cite[Th. 1]{Haghighi20112}.
\end{proof}
\begin{remark}
If we set $u_{k+1}=\overline{\hat{q}_k}$, where $\hat{q}_k$ is an estimate of PU state $q_k$ and $\overline{\cdot}$ denotes negation, $P_0$ is the false-alarm probability and $P_1$ is the probability of missed detection for $\hat{q}_k$.
\end{remark}
\subsection{Bound for the performance of cognitive transmission strategies}
\begin{theorem}\label{theory2}
For any CR that satisfies $\rho\leq \rho_{\text{max}}$,
\begin{align}
\eta\leq \eta_{\text{max}}\triangleq
           \begin{cases}
             \rho_{\text{max}}+(1-a_{01}-a_{10})\min\{\frac{\rho_{\text{max}}}{a_{10}},\frac{1-\rho_{\text{max}}}{a_{11}}\}, & \text{if $a_{01}+a_{10}\leq 1$;} \\
             \rho_{\text{max}}-(1-a_{01}-a_{10})\min\{\frac{\rho_{\text{max}}}{a_{11}},\frac{1-\rho_{\text{max}}}{a_{10}}\}, & \text{if $a_{01}+a_{10}> 1$,}
\end{cases}\label{etaopt}
\end{align}
\end{theorem}
\begin{proof}
Eliminating $1-P_{0}$ from (\ref{SU3}) and (\ref{IR3}) yields
\begin{align}
\eta=\rho+\frac{1-a_{01}-a_{10}}{a_{10}}(\rho-P_{1}).\label{etaofrho}
\end{align}
The feasible range of $P_{1}$  can be calculated from (\ref{IR3}), $0\leq P_{0}\leq 1$ and $0\leq P_{1}\leq 1$ as $\max\{0,\frac{\rho-a_{10}}{a_{11}}\}\leq P_{1} \leq \min\{1,\frac{\rho}{a_{11}}\}$. If $a_{01}+a_{10}\leq 1$, then $\eta$ can be upperbounded by substituting the lower bound on $P_{1}$ and $\rho\leq \rho_{\text{max}}$ in (\ref{etaofrho}), which yields the first line of (\ref{etaopt}). Similarly, if $a_{01}+a_{10}> 1$, then the second line of (\ref{etaopt}) is obtained from (\ref{etaofrho}) and the upper bound on $P_{1}$.
\end{proof}
\begin{corollary}
$\eta_{\text{max}}\geq \rho_{\text{max}}$.
\end{corollary}

\section{Energy detection as baseline CR strategy}\label{energysec}
\indent Energy detection, which is one of the most widely deployed spectrum sensing methods because of its simplicity, compares the estimated received energy ($y_k$) with a threshold to detect the existence or absence of the PU signal. Using this threshold at a certain received PU signal power to CR signal-to-noise ratio (SNR) will result in certain probabilities of mis-detection and false alarm. This procedure is modelled in the HMM presented in Fig. \ref{base}. In this model, $\hat{q}_{k}=0$ and $\hat{q}_{k}=1$ denote the detected state to be $H_0$ and $H_1$, respectively, and thus $\hat{q}_{k}=0$~ if $y_k\leq \theta_e$ or $\hat{q}_{k}=1$~ if $y_k> \theta_e$,
where $\theta_e$ is detection threshold. Thus, $P_{\text{FA}}$ and $P_{\text{M}}$ are
\begin{align}
P_{\text{FA}}&=1-\mathcal{F}_{y_k|q_k}(\theta_e|0)=1-\frac{\gamma(K,\theta_e/\sigma^2_0)}{\Gamma(K)},\label{b00}\\
P_{\text{M}}&=\mathcal{F}_{y_k|q_k}(\theta_e|1)=\frac{\gamma(K,\theta_e/\sigma^2_1)}{\Gamma(K)},\label{b11}
\end{align}
where $\Gamma$ is the Gamma function, $\gamma$ is the lower incomplete Gamma function, $\mathcal{F}_{y_k|q_k}(\cdot|0)$ and $\mathcal{F}_{y_k|q_k}(\cdot|1)$ are the cumulative distribution function (CDF) of a chi-square distribution with $2K$ degrees of freedom and Gaussian variance $\sigma^2_0/2$ and $\sigma^2_1/2$, respectively.\\

\indent We will use $u_{k+1}=\overline{\hat{q}_k}$ as the baseline transmission strategy. The threshold $\theta_e$, that maximizes UR, is calculated by recalling that $P_0=P_{\text{FA}}$, $P_1=P_{\text{M}}$ and combining expressions (\ref{IR3}), (\ref{b00}) and (\ref{b11}), substituting $\rho=\rho_{\text{max}}$ and solving them for $\theta_e$.
\begin{figure}[t]
  \centering
  \psfrag{p}[][][2]{$a_{01}$}
  \psfrag{q}[][][2]{$a_{10}$}
  \psfrag{p2}[][][2]{$a_{00}$}
  \psfrag{q2}[][][2]{$a_{11}$}
  \psfrag{h0}[][][2]{$q_k=0$}
  \psfrag{h1}[][][2]{$q_k=1$}
  \psfrag{a}[][][2]{$1-P_{\text{FA}}$}
  \psfrag{b}[][][2]{$1-P_{\text{M}}$}
  \psfrag{a2}[][][2]{$P_{\text{FA}}$}
  \psfrag{b2}[][][2]{$P_{\text{M}}$}
  \psfrag{v0}[][][2]{$\hat{q}_{k}=0$}
  \psfrag{v1}[][][2]{$\hat{q}_{k}=1$}
  \scalebox{.4}{\includegraphics{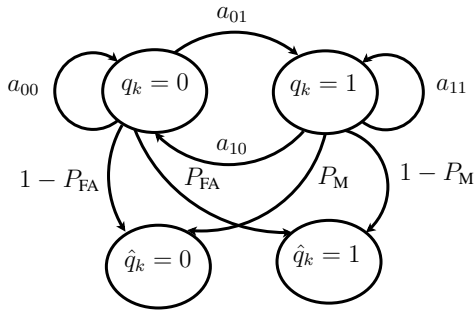}}
  \caption{HMM model for the energy detector.}
 \label{base}
\end{figure}

\section{A-posteriori probabilities LLR based cognitive radio}\label{appsec}
\indent One reasonable way to incorporate both the model and the entire observation is to form the a-posterior probability of $\Pr\{q_{k+1}=1|\mathbf{y}_k\}$. This probability will be used in the decision rule as
\begin{align}
u_{k+1}=           \begin{cases}
             1, & \text{if $z_k\leq \theta_{\text{LLR}}$} \label{CRTX}\\
             0, & \text{if $z_k>\theta_{\text{LLR}}$}
           \end{cases},
    \end{align}
where $z_k\triangleq\log~\frac{\Pr\{q_{k+1}=1|\mathbf{y}_k\}}{\Pr\{q_{k+1}=0|\mathbf{y}_k\}}$ and $\theta_{\text{LLR}}$ are the \emph{a posteriori log-likelihood ratio} and the threshold for $z_k$, respectively. The $z_k$, which is used for estimating the future state of PU, hereafter will be addressed as the LLR.
%
Thus, with the same method explained in \cite[eqs. 18--19]{Haghighi20112}, the LLR as a function of the forward variables $\alpha_k(j)\triangleq\Pr\{q_k=j,\mathbf{y}_k\},~~j\in\{0,1\}$, which are computed recursively\cite[eqs. 19--21]{Rabinner1989} with moderate complexity, is derived as
\begin{align}
z_k=\log \frac{a_{01}\alpha_k(0)+a_{11}\alpha_k(1)}{a_{00}\alpha_k(0)+a_{10}\alpha_k(1)}\label{LLR2}.
\end{align}
\indent In our previous paper\cite{Haghighi20112} the forward variables were calculated based on the discrete output HMM. However, the forward variables can be calculated based on the continuous-output HMM presented in Fig. \ref{xhmm_fix}. There are several benefits in doing the latter. The baseline method in Section~\ref{energysec} needs a threshold to be calculated while the continuous model does not need such a threshold. This thresholding might reduce the information available in the samples from the continuous-output HMM.
Since both $\rho_l(\theta_{\text{LLR}})$ and $\eta_l(\theta_{\text{LLR}})$ are nondecreasing functions of $\theta_{\text{LLR}}$, it follows that the optimum threshold, which does not cause more interference than the allowed $\rho_{\text{max}}$ and maximizes the UR, is found from (\ref{IR1}) as
   \begin{align}
 	\theta_{\text{LLR}}= \mathcal{F}^{-1}_{z_k|q_{k+1}}(\rho_{\text{max}}|1),\label{appllrthresh}
    \end{align}
where $ \mathcal{F}^{-1}_{z_k|q_{k+1}}(\cdot|1)$ is the inverse CDF of $z_k$ conditioned on $q_{k+1}=1$.\\
\indent In the case that the PU transition matrix in (\ref{Amat}) is time-variant, semi-Markov models can be used instead of the model in Fig. \ref{ghmm}. For hidden semi-Markov models, forward variables can be calculated~\cite{Yu2003} and thus the same method can be deployed.
\subsection{Optimality of the LLR based cognitive radio}\label{optimality}
\begin{theorem}\label{theory3}
\indent The a-posteriori LLR-based cognitive transmission scheme presented in (\ref{CRTX}) is the optimum strategy in terms of maximizing UR subject to $\rho\leqslant\rho_{\text{max}}$.
\end{theorem}
\begin{IEEEproof}
\indent The proof is inspired from the proof of the Neyman-Pearson Lemma \cite{neyman1933}. To prove the theorem, it should be shown that for any other strategy $A$, which has $\eta_{\text{A}}$ and $\rho_{\text{A}}\leq\rho_{\text{max}}$, the LLR-based strategy has higher UR $\eta_{\text{LLR}}\geq\eta_{\text{A}}$ with the condition on $\rho_{\text{LLR}}=\rho_{\text{max}}$. The set of observations $\mathbf{Y}_k$ for which the CR decides to transmit is denoted by $R$. Thus, for LLR-based strategy set $R_{\text{LLR}}$ is defined as
\begin{align}
 	R_{\text{LLR}}&\triangleq\left\{\mathbf{y}\in \mathbb{R}^k: \log~\frac{\Pr\{q_{k+1}=1|\mathbf{Y}_k=\mathbf{y}\}}{\Pr\{q_{k+1}=0|\mathbf{Y}_k=\mathbf{y}\}}\leq\theta_{\text{LLR}}\right\}=\left\{\mathbf{y}\in \mathbb{R}^k: \frac{f_{\mathbf{Y}_k|q_{k+1}}(\mathbf{y}|1)}{f_{\mathbf{Y}_k|q_{k+1}}(\mathbf{y}|0)}\leq\theta'_{\text{LLR}}\right\},\notag\\
	\theta'_{\text{LLR}}&=\frac{\pi_0}{\pi_1}e^{\theta_{\text{LLR}}}\notag
	\end{align}
where $f_{\mathbf{Y}_k|q_{k+1}}$ is the distribution of observations given next PU state. The IR and UR can be written as
\begin{align}
\rho=\Pr\left\{\mathbf{Y}_k\in R|q_{k+1}=1\right\}=\int_R f_{\mathbf{Y}_k|q_{k+1}}(\mathbf{y}|1)d\mathbf{y},~
\eta=\Pr\left\{\mathbf{Y}_k\in R|q_{k+1}=0\right\}=\int_R f_{\mathbf{Y}_k|q_{k+1}}(\mathbf{y}|0)d\mathbf{y}.
\end{align}
\indent From law of total probability it can be shown that
\begin{align}
R_{\text{A}}=(R_{\text{A}}\cap R_{\text{LLR}})\cup(R_{\text{A}}\cap R_{\text{LLR}}^c),~R_{\text{LLR}}=(R_{\text{A}}\cap R_{\text{LLR}})\cup(R_{\text{A}}^c\cap R_{\text{LLR}}),
\end{align}
where $R^c$ denotes the complement set of $R$.
 Since the components of the union are disjoint events, the probability that an observation belongs to a set can be written as the sum of the components. Thus, to show that $\eta_{\text{LLR}}\geq\eta_{\text{A}}$, it is enough to show that $\Pr\left\{\mathbf{Y}_k\in R_{\text{A}}^c\cap R_{\text{LLR}}|q_{k+1}=0\right\}\geq\Pr\left\{\mathbf{Y}_k\in R_{\text{A}}\cap R_{\text{LLR}}^c|q_{k+1}=0\right\}$. To prove the theorem, starting from the left side, it can be written
\begin{align}
&\Pr\left\{\mathbf{Y}_k\in R_{\text{A}}^c\cap R_{\text{LLR}}|q_{k+1}=0\right\}=\int_{R_{\text{A}}^c\cap R_{\text{LLR}}} f_{\mathbf{Y}_k|q_{k+1}}(\mathbf{y}|0)d\mathbf{y}\geq \frac{1}{\theta'_{\text{LLR}}}\int_{R_{\text{A}}^c\cap R_{\text{LLR}}}f_{\mathbf{Y}_k|q_{k+1}}(\mathbf{y}|1)d\mathbf{y}\notag\\
&=\frac{1}{\theta'_{\text{LLR}}} \Pr\left\{\mathbf{Y}_k\in R_{\text{A}}^c\cap R_{\text{LLR}}|q_{k+1}=1\right\}=\frac{\rho_{\text{LLR}}-\rho'}{\theta'_{\text{LLR}}}=\frac{\rho_{\text{max}}-\rho'}{\theta'_{\text{LLR}}}
\geq\frac{\rho_{\text{A}}-\rho'}{\theta'_{\text{LLR}}}\notag\\
&=\frac{1}{\theta'_{\text{LLR}}}\Pr\left\{\mathbf{Y}_k\in R_{\text{A}}\cap R_{\text{LLR}}^c|q_{k+1}=1\right\}=\frac{1}{\theta'_{\text{LLR}}}\int_{R_{\text{A}}\cap R_{\text{LLR}}^c}f_{\mathbf{Y}_k|q_{k+1}}(\mathbf{y}|1)d\mathbf{y}\notag\\
&\geq \int_{R_{\text{A}}\cap R_{\text{LLR}}^c}f_{\mathbf{Y}_k|q_{k+1}}(\mathbf{y}|0)d\mathbf{y}=\Pr\left\{\mathbf{Y}_k\in R_{\text{A}}\cap R_{\text{LLR}}^c|q_{k+1}=0\right\},\label{comp_ineq}
\end{align}
where $\rho'=\Pr\left\{\mathbf{Y}_k\in R_{\text{A}}\cap R_{\text{LLR}}|q_{k+1}=1\right\}$. The inequality (\ref{comp_ineq}) is true since
\begin{align}
&\mathbf{y}\in R_{\text{A}}\cap R_{\text{LLR}}^c\Rightarrow\mathbf{y}\in R_{\text{LLR}}^c~\Rightarrow \frac{f_{\mathbf{Y}_k|q_{k+1}}(\mathbf{y}|1)}{f_{\mathbf{Y}_k|q_{k+1}}(\mathbf{y}|0)}\geq\theta'_{\text{LLR}}~\Rightarrow \frac{1}{\theta'_{\text{LLR}}}f_{\mathbf{Y}_k|q_{k+1}}(\mathbf{y}|1)d\mathbf{y}
\geq f_{\mathbf{Y}_k|q_{k+1}}(\mathbf{y}|0).\notag
\end{align}
\end{IEEEproof}

\subsection{Implementation issues}
\indent In this section, the limiting assumptions for using the LLR-based method presented earlier are discussed. By carefully looking at the requirements of the LLR-based method, it is apparent that for calculating the LLRs knowledge of the hidden Markov model is required. In both cases of discrete and continuous-output HMM, the transition matrix $\mathbf{A}$ and the SNR are required. This paper assumes that this information is available or estimated beforehand. In \cite[sec. III-C]{Rabinner1989} the Baum-Welch iterative estimation algorithm, which is equivalent to the well-established expectation-modification (EM) method, is demonstrated. This method will be used to estimate the model parameters from the observations. While examining the performance of the Baum-Welch algorithm is beyond the scope of this paper, there exists a vast amount of literature about its convergence and performance.\\
\indent The second and more challenging issue in the LLR-based method lies in the calculation of the threshold in expression (\ref{appllrthresh}). In this expression, there is a need for the knowledge of the PU states (or their estimates) for a certain training period to estimate $\mathcal{F}_{z_{k}|q_{k+1}}(x|1)$. This is normally done sporadically, but since the true states of PU are not known, they have to be estimated. This process can be done for the previous observations; their corresponding PU states can be estimated with the forward-backward algorithm~\cite{Rabinner1989}. Notice that the estimated states of PU might not perfectly corresponds to the actual ones due to the uncertainties in the noise and channel. This will change the empirical CDF and thus the threshold calculated on which it is based. This error in the PU state estimation will depend deeply on the SNR and also on the $\mathbf{A}$ matrix. The big concern with this error is that it might result in possible violation of the maximum allowed IR for the PU ($\rho_{\text{max}}$). However, to have a useful method, robust to changes and reductions in SNR, it is necessary to make sure that it will never violate the IR under any conditions. In low SNRs in which the PU state estimation might be poor, we can directly use unconditional empirical CDF of LLRs which does not need PU state estimation. In Section~\ref{nopustate}, we proved analytically that the threshold which is calculated based on unconditional CDF of LLRs will result in a CR strategy which does not violate IR. 

\section{Threshold calculation without true PU state knowledge}\label{nopustate}
\indent The threshold for CR transmission strategy can be calculated based on the expression (\ref{appllrthresh}). To do so, the actual PU states are needed to estimate the empirical CDF (ECDF) of LLRs conditioned on PU states. This empirical CDF is used for calculating the decision threshold. In this paper, we estimate the PU states with the forward-backward algorithm. Notice that the scenario where the correct PU states are known is not realistic.\\
\indent In this section, we show that, even without knowing the true state of the PU, it is possible to find a threshold that will not harm the PU. To prove the existence of such threshold, it is sufficient to prove that if the threshold is calculated based on the unconditional empirical CDF, the actual IR will not exceed $\rho_{\text{max}}$. This can be shown by proving that the unconditional CDF of LLRs ($\mathcal{F}_{z_{k}}(x)$) is always bigger than the CDF of LLRs conditioned on the next PU state being one ($\forall x; \mathcal{F}_{z_{k}}(x)\geq \mathcal{F}_{z_{k}|q_{k+1}}(x|1)$). This is proved in theorem \ref{thereshold-theory}. As explained in Section \ref{sys_model}, we focus on the case $a_{01}+a_{10}<1$. The main part of this proof is to show that the empirical CDF of the LLRs conditioned on the next PU state being zero is always larger than the CDF of the LLRs conditioned on the next PU state being one ($\forall x; \mathcal{F}_{z_{k}|q_{k+1}}(x|0)\geq \mathcal{F}_{z_{k}|q_{k+1}}(x|1)$), which is proved in the same theorem. To show this, first it is shown that the $\mathcal{F}_{z_{k}|q_{k+1}}(x|0)\geq \mathcal{F}_{z_{k}|q_{k+1}}(x|1)$ is equivalent to show that $\mathcal{F}_{\Lambda_{k}|q_{k+1}}(x|0)\geq \mathcal{F}_{\Lambda_{k}|q_{k+1}}(x|1)$, where $\Lambda_{k}\triangleq\log \frac{\alpha_k(1)}{\alpha_k(0)}$. Now by inserting the expression for calculating the forward variable \cite[eqs. 19--20]{Rabinner1989} the following expression is obtained
\begin{align}
\Lambda_{k}&= \begin{cases}
             \underbrace{\log \frac{\alpha_{k-1}(0)a_{01}+\alpha_{k-1}(1)a_{11}}{\alpha_{k-1}(0)a_{00}+\alpha_{k-1}(1)a_{10}}}_{z_{k-1}}+\underbrace{\log\frac{b_{1}(y_{k})}{b_{0}(y_{k})}}_{B_{k}}=z_{k-1}+ B_{k}, & \text{if $k>1$} \label{lambda_k}\\
             B_{k}, & \text{if $k=1$}
           \end{cases},
\end{align}
where $b_{i}(\cdot)$ is the probability distribution function of a the Chi-square random variable with $2K$ degrees of freedom and original Gaussian variance of $\sigma^2_{i}/2$. Recall that $\sigma^2_{0}$ is the noise variance and $\sigma^2_{1}$ is the signal plus noise variance $\sigma^2_{1}=\sigma^2_{0}+\sigma^2_{s}$.
\begin{lemma}\label{lemma0}
If $\mathcal{F}_{\Lambda_{k}|q_{k+1}}(x|0)\geq \mathcal{F}_{\Lambda_{k}|q_{k+1}}(x|1)~,~\forall x\in \mathbb{R}$ and $a_{01}+a_{10}<1$ then $\mathcal{F}_{z_{k}|q_{k+1}}(x|0)\geq \mathcal{F}_{z_{k}|q_{k+1}}(x|1)$ for all $x$ in the domain of $z_k$.
\end{lemma}
\begin{IEEEproof}
From (\ref{LLR2}), we have
\begin{align}
z_{k}=\log \frac{a_{01}+a_{11}\frac{\alpha_k(1)}{\alpha_k(0)}}{a_{00}+a_{10}\frac{\alpha_k(1)}{\alpha_k(0)}}= \log \frac{a_{01}+a_{11}e^{\Lambda_{k}}}{a_{00}+a_{10}e^{\Lambda_{k}}}=\log \left[\frac{a_{11}}{a_{10}}-\frac{1-a_{01}-a_{10}}{a_{10}(a_{00}+a_{10}e^{\Lambda_{k}})}\right].\label{lambzk}
\end{align}
\indent Since $1-a_{01}-a_{10}>0$, in (\ref{lambzk}), the second term inside the log has a positive nominator and denominator, and exponential is an increasing function of $\Lambda_{k}$.
Thus, $z_{k}$ is a monotonic increasing function of $\Lambda_{k}$. The lemma follows since the CDFs of $\Lambda_{k}$ and $z_k$ will have the same behaviour.
\end{IEEEproof}


\begin{lemma}\label{lemma1}
\indent For $y_{k}$ as defined in Section \ref{signoise} and $B_k$ defined in (\ref{lambda_k}),
$\mathcal{F}_{B_{k}|q_{k}}(x|0)\geq \mathcal{F}_{B_{k}|q_{k}}(x|1)$
for all $k\geq1$ and $x\geq0$.
\end{lemma}
\begin{IEEEproof}
\indent Starting from derivation of $B_{k}$, we will have \cite[pp. 370]{kendall1973}
\begin{align}
B_{k}=\log\frac{b_{1}(y_{k})}{b_{0}(y_{k})}=\log \frac{\frac{1}{\sigma^{2K}_1 2^{K}\Gamma(K)}y_{k}^{K-1}e^{-y_{k}/2\sigma^2_1}}{\frac{1}{\sigma^{2K}_0 2^{K}\Gamma(K)}y_{k}^{K-1}e^{-y_{k}/2\sigma^2_0}}=2K\log \frac{\sigma_0}{\sigma_1}+\frac{y_{k}}{2}\big(\frac{\sigma^2_1-\sigma^2_0}{\sigma^2_1 \sigma^2_0}\big),\notag
\end{align}
where  $\Gamma(\cdot)$ represents the Gamma function. Because $\sigma^2_1\geq\sigma^2_0$, $B_{k}$ is a strictly increasing function of $y_{k}$. The lemma now follows because
\begin{align}
\mathcal{F}_{y_{k}|q_{k}}(y|1)=\frac{\int_0^{y/2\sigma^2_1}t^{K-1}e^{-t}dt}{\Gamma(K)}\leq \frac{\int_0^{y/2\sigma^2_0}t^{K-1}e^{-t}dt}{\Gamma(K)}=\mathcal{F}_{y_{k}|q_{k}}(y|0).\notag
\end{align}

\end{IEEEproof}

\begin{lemma}\label{lemmanew}
Let $C_k$ be any stationary random process that conditioned on $q_k$ is independent of $q_{k+1}$. If $a_{01}+a_{10}<1$,  then for any $x$,
\begin{align}
\mathcal{F}_{C_{k}|q_{k}}(x|0)\geq \mathcal{F}_{C_{k}|q_{k}}(x|1)~\Leftrightarrow~\mathcal{F}_{C_{k}|q_{k+1}}(x|0)\geq \mathcal{F}_{C_{k}|q_{k+1}}(x|1).\label{C_ineq}
\end{align}
\end{lemma}
\begin{IEEEproof}
From the conditional independence in the lemma assumption we have $\Pr\{C_k\leq x|q_k=i, q_{k+1}=j\}=\Pr\{C_k\leq x|q_k=i\}$. Now, for $i\in\{0,1\}$ and $j\in \{0,1\}$
\begin{align}
\Pr\{C_k\leq x,q_k&=i,q_{k+1}=j\}=\Pr\{C_k\leq x|q_k=i,q_{k+1}=j\}\Pr\{q_k=i,q_{k+1}=j\}\notag\\
&=\Pr\{C_k\leq x|q_k=i\}\Pr\{q_{k+1}=j|q_k=i\}\Pr\{q_k=i\}=\mathcal{F}_{C_{k}|q_{k}}(x|i)a_{ij}\pi_i\notag.\\
\Pr\{C_{k}\leq x|q_{k+1}=j\}&=\frac{p_{0}a_{0j}\pi_{0}+p_{1}a_{1j}\pi_{1}}{\pi_{j}}.\notag
\end{align}
where $p_{0}=\mathcal{F}_{C_{k}|q_{k}}(x|0)$ and $p_{1}=\mathcal{F}_{C_{k}|q_{k}}(x|1)$. Now since, by assumption, $1-a_{10}-a_{01} = a_{11} - a_{01} = a_{00}-a_{10} > 0$, we have that
\begin{align}
&\mathcal{F}_{C_{k}|q_{k}}(x|0)\geq \mathcal{F}_{C_{k}|q_{k}}(x|1)~\Leftrightarrow ~ p_{0}(a_{00}-a_{10})\geq p_{1}(a_{11}-a_{01})\notag\\
&\Leftrightarrow \frac{p_{0}a_{00}a_{10}+p_{1}a_{10}a_{01}}{a_{10}}\geq\frac{p_{0}a_{01}a_{10}+p_{1}a_{01}a_{11}}{a_{01}}\Leftrightarrow ~\Pr\{C_{k}\leq x|q_{k+1}=0\}\geq\Pr\{C_{k}\leq x|q_{k+1}=1\}\notag.
\end{align}
\end{IEEEproof}
In Lemma \ref{lemma1} it was proved that $\mathcal{F}_{B_{k}|q_{k}}(x|0)\geq \mathcal{F}_{B_{k}|q_{k}}(x|1)$. Also $B_{k}$ conditioned on $q_{k}$ is independent of $q_{k+1}$, which yields the following corollary.
\begin{corollary}\label{bqkp1}
If $a_{01}+a_{10}< 1$ then $\mathcal{F}_{B_{k}|q_{k+1}}(x|0)\geq \mathcal{F}_{B_{k}|q_{k+1}}(x|1)$.
\end{corollary}
\begin{lemma}\label{akcond}
If $\mathcal{F}_{\Lambda_{k}|q_{k+1}}(x|0)\geq \mathcal{F}_{\Lambda_{k}|q_{k+1}}(x|1)$ and $a_{01}+a_{10}<1$ then
$\mathcal{F}_{z_k|q_{k}}(x|0)\geq \mathcal{F}_{z_{k}|q_{k}}(x|1)$.
\end{lemma}
\begin{IEEEproof}
From the assumptions made in this lemma and Lemma \ref{lemma0}, $\mathcal{F}_{z_k|q_{k+1}}(x|0)\geq \mathcal{F}_{z_{k}|q_{k+1}}(x|1)$. Now, since the $z_{k}$ fulfils the properties specified for $C_k$ in lemma \ref{lemmanew}, this lemma follows.
\end{IEEEproof}

\begin{lemma}\label{lemmamain}
If $\mathcal{F}_{z_{k}|q_{k}}(x|0)\geq \mathcal{F}_{z_{k}|q_{k}}(x|1)$
and $a_{01}+a_{10}<1$ then
\begin{align}
\mathcal{F}_{z_{k}+B_{k+1}|q_{k+2}}(x|0)\geq \mathcal{F}_{z_{k}+B_{k+1}|q_{k+2}}(x|1).\label{main_ineq}
\end{align}
\end{lemma}
\begin{IEEEproof}
\indent Starting from Lemma \ref{lemma1} we will have $\mathcal{F}_{B_{k+1}|q_{k+1}}(x|0)\geq \mathcal{F}_{B_{k+1}|q_{k+1}}(x|1)$.  Since the states $q_k$ form a Markov chain, the dependences between $z_k$, $B_{k+1}$, and $q_{k+2}$ are depicted as
\begin{displaymath}
    \xymatrix{
        \cdots \ar[dr] \ar[r] & q_k \ar[d] \ar[r] & q_{k+1} \ar[d] \ar[r] &q_{k+2}\\
          & z_k & B_{k+1}                        }.
\end{displaymath}
\indent Thus, using the chain rule and Markov property, the joint distribution can be written as \cite[pp. 37-38]{henkbook2007}
\begin{align}
\Pr\{z_k+B_{k+1}\leq x, q_k,q_{k+1},q_{k+2}\}=\Pr\{q_k\}\Pr\{q_{k+1}|q_{k}\}\Pr\{q_{k+2}|q_{k+1}\}\Pr\{z_k+B_{k+1}\leq x|q_k,q_{k+1}\}.\label{joint_all}
\end{align}
\indent On the other hand, the CDF of the sum of two independent r.vs $A$ and $B$ can be expressed as \cite[pp. 187--190]{kendall1973}
\begin{align}
\mathcal{F}_{A+B}(x)=\mathcal{F}_{A}(x)\ast f_{B}(x)=\mathcal{F}_{B}(x)\ast f_{A}(x),\label{conv}
\end{align}
where $f_{A}(\cdot)$ is the PDF of $A$ and $\ast$ denotes convolution.\\
\indent Since $z_k$ depends only on $q_k$ and the previous states (and channel noise which is independent of the PU states) and $B_{k+1}$ depends solely on $q_{k+1}$ (and noise), the sum of them conditioned on $q_k,q_{k+1}$ can be written as
\begin{align}
\mathcal{F}_{z_k+B_{k+1}|q_k,q_{k+1}}(x|i,j)
=f_{z_k|q_k,q_{k+1}}(x|i,j)\ast \mathcal{F}_{B_{k+1}|q_k,q_{k+1}}(x|i,j)
=f_{z_k|q_k}(x|i)\ast \mathcal{F}_{B_{k+1}|q_{k+1}}(x|j) .\label{cond_sum}
\end{align}
\indent To derive both sides of the inequality (\ref{main_ineq}), one should marginalize the joint distribution in (\ref{joint_all}) with respect to $q_k$ and $q_{k+1}$ and divide it to the $\Pr\{q_{k+2}=i\},i\in\{0,1\}$. After doing that and plugging (\ref{cond_sum}) in (\ref{joint_all}), for the l.h.s and r.h.s of (\ref{main_ineq}) we will have, respectively
\begin{align}
\mathcal{F}_{z_{k}+B_{k+1}|q_{k+2}}(x|0)&=a_{00}^2\mathcal{A}'_0\ast \mathcal{B}_0+a_{01}a_{10}\mathcal{A}'_0\ast \mathcal{B}_1+a_{01}a_{00}\mathcal{A}'_1\ast \mathcal{B}_0+a_{01}a_{11}\mathcal{A}'_1\ast \mathcal{B}_1,\label{sum0}\\
\mathcal{F}_{z_{k}+B_{k+1}|q_{k+2}}(x|1)&=a_{10}a_{00}\mathcal{A}'_0\ast \mathcal{B}_0+a_{11}a_{10}\mathcal{A}'_0\ast \mathcal{B}_1+a_{10}a_{01}\mathcal{A}'_1\ast \mathcal{B}_0+a_{11}^2\mathcal{A}'_1\ast \mathcal{B}_1,\label{sum1}
\end{align}
where $\mathcal{A}_i=\mathcal{F}_{z_k|q_k}(x|i)$, $\mathcal{B}_i=\mathcal{F}_{B_{k+1}|q_{k+1}}(x|i)$, $\mathcal{A}'_i=f_{z_k|q_k}(x|i)$ and $\mathcal{B}'_i=f_{B_{k+1}|q_{k+1}}(x|i)$.

By multiplying both sides of $\mathcal{A}_0\geq \mathcal{A}_1$ with the positive value $1-a_{01}-a_{10}$ and rearranging it, we obtain $a_{00}\mathcal{A}_{0}+a_{01}\mathcal{A}_{1}\geq a_{10}\mathcal{A}_{0}+a_{11}\mathcal{A}_{1}$. Now if both sides of this inequality are convolved with the positive function $\mathcal{B}'_0$, we arrive at
$
(a_{00}\mathcal{A}_{0}+a_{01}\mathcal{A}_{1})\ast \mathcal{B}'_{0}\geq(a_{10}\mathcal{A}_{0}+a_{11}\mathcal{A}_{1})\ast \mathcal{B}'_{0}$.
Now from (\ref{conv}) we can rewrite it as
 $(a_{00}\mathcal{A}'_{0}+a_{01}\mathcal{A}'_{1})\ast \mathcal{B}_{0}\geq(a_{10}\mathcal{A}'_{0}+a_{11}\mathcal{A}'_{1})\ast \mathcal{B}_{1}$ where the last inequality follows because $a_{10}\mathcal{A}'_{0}+a_{11}\mathcal{A}'_{1}\geq0$ and $\mathcal{B}_{0}\geq \mathcal{B}_{1}$ from Lemma \ref{lemma1}. Finally, after multiplying both sides of previous inequality with the positive value of $1-a_{01}-a_{10}$ we get
\begin{align}
(a_{00}\mathcal{A}'_0+a_{01}\mathcal{A}'_1)(a_{00}-a_{10})\ast \mathcal{B}_0\geq(a_{10}\mathcal{A}'_0+a_{11}\mathcal{A}'_1)(a_{11}-a_{01})\ast \mathcal{B}_1\Rightarrow\notag\\
a_{00}^2\mathcal{A}'_0\ast \mathcal{B}_0+a_{01}a_{10}\mathcal{A}'_0\ast \mathcal{B}_1
+a_{01}a_{00}\mathcal{A}'_1\ast \mathcal{B}_0+a_{01}a_{11}\mathcal{A}'_1\ast \mathcal{B}_1\notag\\
\geq a_{10}a_{00}\mathcal{A}'_0\ast \mathcal{B}_0+a_{11}a_{10}\mathcal{A}'_0\ast \mathcal{B}_1
+a_{10}a_{01}\mathcal{A}'_1\ast \mathcal{B}_0+a_{11}^2\mathcal{A}'_1\ast \mathcal{B}_1,\Rightarrow\notag\\
\mathcal{F}_{z_{k}+B_{k+1}|q_{k+2}}(x|0)\geq \mathcal{F}_{z_{k}+B_{k+1}|q_{k+2}}(x|1),\notag
\end{align}
where the last step follows from (\ref{sum0}) and (\ref{sum1}).
\end{IEEEproof}

\begin{theorem}\label{thereshold-theory}
\indent If $\theta'=\mathcal{F}_{z_{k}}^{-1}(\rho_{\text{max}})$ and  $a_{01}+a_{10}< 1$, then $\mathcal{F}_{z_{k}|q_{k+1}}(\theta'|1)\leq \rho_{\text{max}}$.
\end{theorem}
\begin{IEEEproof}
\indent From Lemma {\ref{lemma0}} $\mathcal{F}_{z_{k}|q_{k+1}}(x|0)\geq \mathcal{F}_{z_{k}|q_{k+1}}(x|1)$ is the same as proving that $\mathcal{F}_{\Lambda_{k}|q_{k+1}}(x|0)\geq \mathcal{F}_{\Lambda_{k}|q_{k+1}}(x|1)$. To do so, induction is used. First, $\mathcal{F}_{\Lambda_{1}|q_{2}}(x|0)\geq \mathcal{F}_{\Lambda_{1}|q_{2}}(x|1)$ for all $x$ by (\ref{lambda_k}) and corollary \ref{bqkp1}. Second, Lemma \ref{akcond} and Lemma \ref{lemmamain} show that if $\mathcal{F}_{\Lambda_{k}|q_{k+1}}(x|0)\geq \mathcal{F}_{\Lambda_{k}|q_{k+1}}(x|1)$ for any $k\geq 1$ and any $x$ then $\mathcal{F}_{\Lambda_{k+1}|q_{k+2}}(x|0)\geq \mathcal{F}_{\Lambda_{k+1}|q_{k+2}}(x|1)$, which completes the induction. Hence $\mathcal{F}_{z_{k}|q_{k+1}}(x|0)\geq \mathcal{F}_{z_{k}|q_{k+1}}(x|1)$ for any $k\geq 1$ and any $x$. Now from the assumption about $\rho_{\text{max}}$
\begin{align}
\rho_{\text{max}}&=\mathcal{F}_{z_k}(\theta')=\pi_0\mathcal{F}_{z_k|q_{k+1}}(\theta'|0)+\pi_1\mathcal{F}_{z_k|q_{k+1}}(\theta'|1)\notag\\
&\geq\pi_0\mathcal{F}_{z_k|q_{k+1}}(\theta'|1)+\pi_1\mathcal{F}_{z_k|q_{k+1}}(\theta'|1)=\mathcal{F}_{z_{k}|q_{k+1}}(\theta'|1).\notag
\end{align}
\end{IEEEproof}
\begin{corollary}
If $a_{01}+a_{10}< 1$ then for LLR-based CR strategy $\eta\geq \rho$.
\end{corollary}

\indent Thus, the CR strategy with a threshold found based on the unconditional CDF of all LLRs protects the PU ($\rho\leq\rho_{\text{max}}$). One assumption which has been made in most of the lemmas and Theorem in this section the requirement is to have $a_{01}+a_{10}< 1$. Since in the system model we assumed that the CR slot length is much smaller than the PU slot length, the probability of transition from one state to another will be small. Thus, having $a_{01}+a_{10}< 1$ is not a heavy assumption and can be realized easily in practice.
\section{Performance evaluation and results}\label{results}
\indent We compare the LLR-based strategy with three different methods for calculating the threshold with the classical energy detection based spectrum sensing described in Section \ref{energysec}. For all of these simulations, the same PU Markov model ($\mathbf{A}$) and same level of interference $\rho_{\text{max}}$ is used.\\
\indent The threshold needed for the LLR method is calculated by
replacing $\mathcal{F}_{z_k \mid q_{k+1}}(x\mid 1)$ in
(\ref{appllrthresh}) 
with an empirical (sample) CDF. The empirical CDF is computed from the
set of training data $\mathcal{Z}_T = \{z_1, z_2, \ldots, z_{N_T}\}$,
where $N_T$ is assumed to be large enough such that the empirical CDF
is a close approximation of the corresponding CDF. In this paper, we
compute the empirical CDF from one the following three subsets of
$\mathcal{Z}_T$,
\begin{enumerate}[(i)]
\item $\{z_k \in \mathcal{Z}_T : q_{k+1} = 1\}$, i.e, when
  the PU states are assumed to be known
\item $\{z_k \in \mathcal{Z}_T : \hat q_{k+1} = 1\}$,
  where $\hat q_{k+1}$ is the estimated PU states from the
  forward-backward method
\item $\mathcal{Z}_T$, i.e.,
  the ECDF is a close estimate of the unconditional CDF of $z_k$
\end{enumerate}
Note that method (i) is unrealistic, while (ii) and (iii) are more
practical for calculating the threshold.
 The rest of this section discusses the evaluation setup by which these CRs are assessed. It then presents some results and a comparison.
\subsection{Evaluation setup}
\indent In simulating the performance of a CR transmission strategy, the ratio of received primary signal power (at the CR receiver) to the CR receiver noise power is important. For the sake of  simplicity, we assume one PU link and one CR link. It might be possible to extend it to a case with multiple coordinated PUs and multiple coordinated CRs. Moreover, we define the SNR as $\text{SNR}\triangleq\sigma^2_s/\sigma^2_0$ (in dB). In this simulation, $K$ is selected to be $10$.  This parameter plays a role for the SNR scaling. The other factor which is important in evaluating CRs is the maximum allowable IR $\rho_{\text{max}}$. This parameter is normally decided by regulatory bodies like the FCC. In simulations, $\rho_{\text{max}}$ is chosen to be 10\%. The number of elements in $\mathcal{Z}_T$ is $N_T=5\cdot10^{-6}$. To evaluate the performance another $5\cdot10^{-6}$ slots are simulated.
\subsection{Results}
The UR and IR of the different CRs are plotted versus SNR in Fig.~\ref{CHMM1}
and \ref{CHMM2}. The thresholds for the LLR-methods are computed using the
methods (i), (ii), and (iii) described above.
For simplicity of the
discussion, we assume that all ECDFs are close approximations to the
corresponding CDFs. We recall that method~(i) gives an optimum
threshold (i.e., maximizing UR while keeping IR no larger than
$\rho_{\max}$) and that method~(iii) will give a threshold that
guarantees that IR does not exceed $\rho_{\max}$. For method~(ii), we
have no guarantees for the IR.

As expected, the UR of method~(i) is monotonically increasing with SNR
and will approach the upper bound~(\ref{etaopt}) for high SNRs and $\rho_{\max}$
for low SNRs in both Fig.~\ref{CHMM1} and \ref{CHMM2}. In all cases, the UR of method~(i)
is greater or equal to that of the baseline method. However, the UR
and IR curves for methods~(ii) and (iii) behave quite differently in
Fig.~\ref{CHMM1} and \ref{CHMM2}. We note that one important difference between the
simulation setups is that $\pi_0 < \pi_1$ in Fig.~4 and $\pi_0 >
\pi_1$ in Fig.~5, and this will allow us to explain the behavior of
methods~(ii) and (iii).

Let us start with method (ii), which estimates the PU states using the
forward-backward method in the
training phase. In Fig.~\ref{CHMM1}, the UR is very close to the
optimum UR for all considered SNRs and the IR is not exceeding
$\rho_{\max}$. However, in Fig.~\ref{CHMM2}, the performance is close to optimum
only for high SNRs. For low SNRs, the IR exceeds
$\rho_{\max}$, and the CR is in clear violation of the IR
requirement. To explain the different low-SNR behaviors, we recall that
as the SNR approach 0 (in linear scale), the observation $y_1, \ldots,
y_{N_T}$ becomes irrelevant to the PU state estimation. Indeed, as $\mathrm{SNR} \rightarrow 0$, $\hat{q}_{k+1}$ converges in
probability to 1 if $\pi_1>\pi_0$ and 0 if $\pi_1<\pi_0$, for all
$k=1,2,\ldots,N_T$. This implies that $\{z_k \in \mathcal{Z}_T :
\hat{q}_{k+1}=1 \}$ converges to $\mathcal{Z}_T$ if $\pi_1>\pi_0$ and
$\varnothing$ if $\pi_1<\pi_0$ Hence, if $\pi_1<\pi_0$, which is the
case in Fig.~\ref{CHMM2}, we expect method (ii) to completely fail as the SNR
tends to 0. The numerical results in Fig.~5 further indicates that for
low SNRs, method~(ii) will give a too high threshold, resulting in an
IR violation (we cannot estimate the IR and UR reliably for method (ii) at SNRs below $-10$dB with this simulation length, since the training set then is empty with high probability). Conversely, if $\pi_1 > \pi_0$, method~(ii) will
approach method~(iii) as the SNR approach 0. This implies that for
very low SNRs, method~(ii) will not result in an IR violation and that
the UR will be similar to that of method~(iii). This reasoning is
consistent with the results in Fig.~\ref{CHMM1}.

We can conclude that method (ii) is close to optimum for all SNRs when
$\pi_1$ is significantly larger than $\pi_0$. If $\pi_1$ is
significantly smaller than $\pi_0$, then the method works close to
optimum only for SNRs above a certain critical SNR. Below the critical
SNR, the method leads to IR violations, and the method is therefore
invalid in this regime. Continuing with method (iii), we recall that the threshold for this
method, $\theta$, is such that $\mathcal{F}_{z_k}(\theta)
=\rho_{\max}$ and that
the unconditional CDF can be written as
$
\mathcal{F}_{z_k}(x) =
\mathcal{F}_{z_k\mid q_{k+1}}(x\mid 0)\pi_0
+ \mathcal{F}_{z_k\mid q_{k+1}}(x\mid 1)\pi_1
$.
Hence, if $\pi_1\to 1$ then $\mathcal{F}_{z_k}(x) \to
\mathcal{F}_{z_k\mid q_{k+1}}(x\mid 1)$, which implies that
$\rho_{\max} = \mathcal{F}_{z_k}(\theta) \to \mathcal{F}_{z_k\mid
  q_{k+1}}(\theta\mid 1)$. Now, since $\rho_{\max} = \mathcal{F}_{z_k\mid
  q_{k+1}}(\theta^*\mid 1)$ is satisfied for the optimum threshold,
$\theta^*$, it
follows that the UR of method (iii) will be close to optimum.
 Now, in Fig.~\ref{CHMM1}, $\pi_1 = 0.91$ and there will therefore be a gap between
the UR for method~(iii) and the optimum method. Conversely, if $\pi_0\to 1$ then $\mathcal{F}_{z_k}(x) \to
\mathcal{F}_{z_k\mid q_{k+1}}(x\mid 0)$, which implies that
$\rho_{\max} = \mathcal{F}_{z_k}(\theta) \to \mathcal{F}_{z_k\mid
  q_{k+1}}(\theta\mid 0) = \eta$. Hence, the UR for method~(iii) tends
to $\rho_{\max}$. In Fig.~5, $\pi_0 = 0.91$ and there is therefore a
slight gap between the UR for method~(iii) and $\rho_{\max}$. From this we conclude that method~(iii) works best when $\pi_1$ is
large. For the case when $\pi_0$ is large, the threshold is too
conservative resulting in a large UR penalty. However, the IR is
never violated and method~(iii) is the only practical method that is
valid for low SNR when $\pi_0$ is large.
\begin{figure}[t]
  \centering
  \psfrag{snr}[][][1.7]{SNR (dB)}
  \psfrag{eta}[][][1.7]{$\eta$}
  \psfrag{rho}[][][1.7]{$\rho$}
  \scalebox{.6}{\includegraphics{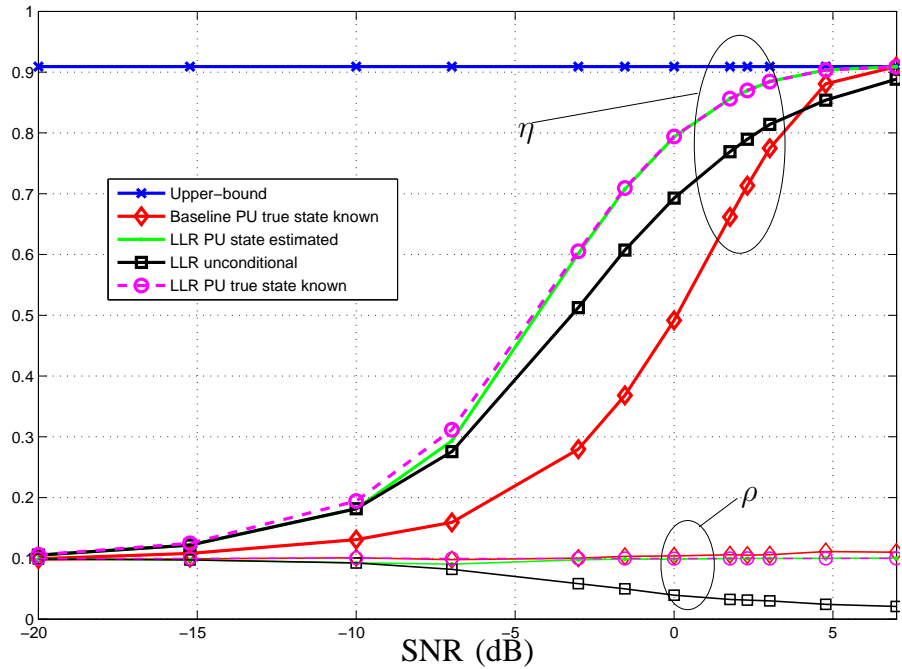}}
  \caption{UR (thick lines) and IR (thin lines) vs. SNR for the baseline CR and corresponding continuous HMM LLR-based CR at $\rho_{\text{max}}=10\%$, $a_{01}=0.1$ and $a_{10}=0.01$}
 \label{CHMM1}
\end{figure}
\begin{figure}[t]
  \centering
  \psfrag{snr}[][][1.2]{SNR (dB)}
  \scalebox{.8}{\includegraphics{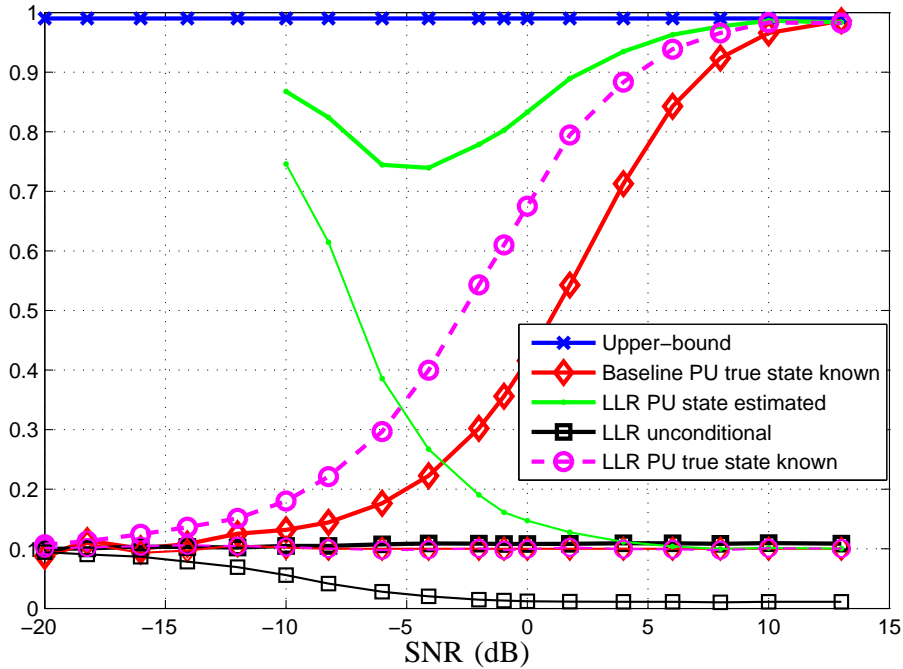}}
  \caption{UR (thick lines) and IR (thin lines) vs. SNR for the baseline CR and corresponding continuous HMM LLR-based CR at $\rho_{\text{max}}=10\%$, $a_{01}=0.01$ and $a_{10}=0.1$}
 \label{CHMM2}
\end{figure}

\section{Conclusion}\label{conc}
\indent In this paper, we have introduced a framework that models the PU, channel, and CR receiver front-end with a simple two-state, continuous-output HMM. The CR transmission strategy can, in general, be viewed as computing a decision variable from the HMM output and comparing the decision variable with a threshold. Hence, to specify a CR transmission strategy, we need only to specify the how to compute the decision variable and how to set the threshold. The performance of a transmission strategy is measured by its UR, under the constraint that the IR does not exceed $\rho_{\max}$. In Theorem \ref{theory2}, we proved an upper bound on the UR, which is a function of the HMM model parameters and $\rho_{\max}$. Theorem~\ref{theory3} states that the optimum decision variable is the APP LLR $z_k$, as defined in~(\ref{LLR2}). The LLRs can be computed from the forward variables, which, in turn, can be computed with moderate complexity\cite{Rabinner1989}. Numerical results show that using the LLR decision variable gives large gains compared to the baseline method, which is based on simple energy detection. The gains are due to the fact that the LLR method make use of all past observations of the PU activity and knowledge of the HMM parameters.\\
\indent It is easy to show that both the UR and the IR are nondecreasing functions of the threshold. Hence, the optimum threshold, i.e., the threshold that maximizes the UR under the constraint that the IR is less or equal to $\rho_{\max}$, is therefore the largest threshold that satisfies the IR constraint with equality. However, to compute the optimum threshold from the CDF of $z_k$ conditioned on that the future PU state $q_{k+1}=1$ is problematic since $q_{k+1}$ is not observable.  The obvious method of (a) estimating the PU states during a training period with the forward backward method, (b) estimating the conditional CDF with an empirical CDF, and (c) finding the threshold from the ECDF and $\rho_{\max}$, is numerically shown to be very close to optimum for all considered SNRs when the PU activity level is high, i.e., when the probability of PU transmission is high. In the opposite situation of a low PU activity level, the method is still close to optimum above a certain SNR, but fails for low SNRs in that the IR exceeds $\rho_{\max}$. A method as the above, but based on the (unconditional) ECDF for $z_k$, obviously avoids the need to estimate the PU states. Furthermore, this method is proven in Theorem \ref{thereshold-theory} to never violate $\rho_{\max}$, regardless of SNR and PU activity levels, but under certain conditions on the PU state transition probabilities, which are argued to be satisfied in practice. Numerical results show that the method works reasonably well when the PU activity level is high. However, the UR is very low compared to the optimum scheme when the PU activity level is low and the SNR is high.\\
\indent In summary, the paper presents practical methods for computing close to optimum thresholds in all cases, except when the SNR and the PU activity level are both low. In the latter case, we can still compute a threshold that respects $\rho_{\max}$, but with a significant loss in UR compared what is achievable with the optimum method. As an example of the former situation with a high PU activity level, our simulation showed of a 116\% UR gain compared to the baseline method at an SNR of $-3$~dB and maximum IR level of 10\%, when the LLR threshold was computed from estimated PU states.
\bibliographystyle{IEEEtran}
\bibliography{IEEEabrv,ref}

\begin{thebibliography}{10}
\providecommand{\url}[1]{#1}
\csname url@samestyle\endcsname
\providecommand{\newblock}{\relax}
\providecommand{\bibinfo}[2]{#2}
\providecommand{\BIBentrySTDinterwordspacing}{\spaceskip=0pt\relax}
\providecommand{\BIBentryALTinterwordstretchfactor}{4}
\providecommand{\BIBentryALTinterwordspacing}{\spaceskip=\fontdimen2\font plus
\BIBentryALTinterwordstretchfactor\fontdimen3\font minus
  \fontdimen4\font\relax}
\providecommand{\BIBforeignlanguage}[2]{{%
\expandafter\ifx\csname l@#1\endcsname\relax
\typeout{** WARNING: IEEEtran.bst: No hyphenation pattern has been}%
\typeout{** loaded for the language `#1'. Using the pattern for}%
\typeout{** the default language instead.}%
\else
\language=\csname l@#1\endcsname
\fi
#2}}
\providecommand{\BIBdecl}{\relax}
\BIBdecl

\bibitem{dohler2011}
M.~Dohler, R.~Heath, A.~Lozano, C.~Papadias, and R.~Valenzuela, ``Is the {PHY}
  layer dead?'' \emph{{IEEE} Commun. Mag.}, vol.~49, no.~4, pp. 159--165, Apr.
  2011.

\bibitem{FCC_spectrum}
{\relax Federal Communications Commission}, ``Spectrum policy task force
  report,'' Federal Communications Commission, Tech. Rep. ET Docket No. 02-155,
  Nov. 2002.

\bibitem{Molisch2009}
A.~Molisch, L.~Greenstein, and M.~Shafi, ``Propagation issues for cognitive
  radio,'' \emph{Proc. {IEEE}}, vol.~97, no.~5, pp. 787--804, May 2009.

\bibitem{Haykin2005}
S.~Haykin, ``Cognitive radio: {Brain}-empowered wireless communications,''
  \emph{{IEEE} J. Sel. Areas Commun.}, vol.~23, no.~2, pp. 201--220, Feb. 2005.

\bibitem{Mitola1993}
J.~Mitola~III, ``Software radios: Survey, critical evaluation and future
  directions,'' \emph{{IEEE} Aerosp. Electron. Syst. Mag.}, vol.~8, no.~4, pp.
  25--36, Apr. 1993.

\bibitem{Yecek2009}
T.~Yucek and H.~Arslan, ``A survey of spectrum sensing algorithms for cognitive
  radio applications,'' \emph{{IEEE} Commun. Surveys Tuts.}, vol.~11, no.~1,
  pp. 116--130, 2009.

\bibitem{Poor1994}
H.~V. Poor, \emph{An Introduction to Signal Detection and Estimation}.\hskip
  1em plus 0.5em minus 0.4em\relax New York: Springer, 1994.

\bibitem{Haddad2007}
M.~Haddad, A.~M. Hayar, and M.~Debbah, ``Spectral efficiency of cognitive radio
  systems,'' in \emph{Proc. {IEEE} Global Telecommunications Conference
  ({Globecom})}, Washington, DC, USA, Nov. 2007, pp. 4165--4169.

\bibitem{Devroye2006}
N.~Devroye, P.~Mitran, and V.~Tarokh, ``Achievable rates in cognitive radio
  channels,'' \emph{{IEEE} Trans. Inf. Theory}, vol.~52, no.~5, pp. 1813--1827,
  May 2006.

\bibitem{tandra2008}
R.~Tandra and A.~Sahai, ``{SNR} walls for signal detection,'' \emph{{IEEE} J.
  Sel. Topics Signal Process.}, vol.~2, no.~1, pp. 4--17, Feb. 2008.

\bibitem{haghighi2011}
M.~Rashidi, K.~Haghighi, A.~Owrang, and M.~Viberg, ``A wideband spectrum
  sensing method for cognitive radio using {sub-Nyquist} sampling,'' in
  \emph{Proc. Digital Signal Processing Workshop and {IEEE} Signal Processing
  Education Workshop {{DSP/SPE}}}, Sedona, Arizona, USA, Jan. 2011.

\bibitem{haghighi20113}
M.~Rashidi, K.~Haghighi, A.~Panahi, and M.~Viberg, ``An {NLLS} based
  {sub-Nyquist} rate spectrum sensing for wideband cognitive radio,'' in
  \emph{Proc. IEEE International Symposium on Dynamic Spectrum Access Networks
  (DySPAN)}, Aachen, Germany, May 2011, pp. 545--551.

\bibitem{Akyildiz2006}
I.~F. Akyildiz, W.-Y. Lee, M.~C. Vuran, and S.~Mohanty, ``Next
  generation/dynamic spectrum access/cognitive radio wireless networks: A
  survey,'' \emph{Computer Networks}, vol.~50, no.~13, pp. 2127--2159, 2006.

\bibitem{Chou2007}
C.-T. Chou, N.~Sai~Shankar, H.~Kim, and K.~Shin, ``What and how much to gain by
  spectrum agility?'' \emph{{IEEE} J. Sel. Areas Commun.}, vol.~25, no.~3, pp.
  576--588, Apr. 2007.

\bibitem{poor2009}
H.~V. Poor and O.~Hadjiliadis, \emph{Quickest Detection}.\hskip 1em plus 0.5em
  minus 0.4em\relax Cambridge University Press, 2009.

\bibitem{haghighi2009}
D.~Noguet, K.~Haghighi, Y.~A. Demessie, L.~Biard, A.~Bouzegzi, M.~Debbah,
  P.~Jallon, M.~Laugeois, P.~Marques, M.~Murroni, J.~Palicot, C.~Sun,
  S.~Thilakawardana, and A.~Yamaguchi, ``Sensing techniques for cognitive
  radio-state of the art and trends,'' IEEE, White Paper SCC41 P1900.6, Apr.
  2009.

\bibitem{haghighi2010}
K.~Haghighi, A.~Svensson, and E.~Agrell, ``Wideband sequential spectrum sensing
  with varying thresholds,'' in \emph{Proc. {IEEE} Global Telecommunications
  Conference ({Globecom})}, Miami, Florida, USA, Dec. 2010.

\bibitem{Rabinner1989}
L.~Rabiner, ``A tutorial on hidden {Markov} models and selected applications in
  speech recognition,'' \emph{Proc. {IEEE}}, vol.~77, no.~2, pp. 257--286, Feb.
  1989.

\bibitem{ChenZ2009}
Z.~Chen, Z.~Hu, and R.~Qiu, ``Quickest spectrum detection using hidden {Markov}
  model for cognitive radio,'' in \emph{Proc. {IEEE} Military Communications
  Conference ({MILCOM})}, Boston, USA, Oct. 2009.

\bibitem{ghosh2010}
C.~Ghosh and D.~P. Agrawal, \emph{Spectrum selection, sensing, and sharing in
  cognitive radio networks : game theoretic approaches and hidden {Markov}
  models to spectrum sensing and sharing}.\hskip 1em plus 0.5em minus
  0.4em\relax Saarbr\"{u}cken, Germany: Lambert Academic Publishing, 2010.

\bibitem{Coulson2009}
A.~Coulson, ``Spectrum sensing using hidden {Markov} modeling,'' in \emph{Proc.
  {IEEE} International Conference on Communications ({ICC})}, Cape Town, South
  Africa, Jun. 2009.

\bibitem{Akbar2007}
I.~Akbar and W.~Tranter, ``Dynamic spectrum allocation in cognitive radio using
  hidden {Markov} models: Poisson distributed case,'' in \emph{Proc. {IEEE}
  SoutheastCon}, Richmond, VA, USA, Mar. 2007.

\bibitem{XiaoLZ09}
X.~Xiao, K.~Liu, and Q.~Zhao, ``Opportunistic spectrum access in self-similar
  primary traffic,'' \emph{EURASIP J. Adv. Sig. Proc.}, vol. 2009, 2009.

\bibitem{Li2009}
H.~Li, ``Restless watchdog: Selective quickest spectrum sensing in multichannel
  cognitive radio systems,'' \emph{EURASIP J. Adv. Sig. Proc.}, vol. 2009,
  2009.

\bibitem{zhang2011}
Z.~Zhang, Z.~Han, H.~Li, D.~Yang, and C.~Pei, ``Belief propagation based
  cooperative compressed spectrum sensing in wideband cognitive radio
  networks,'' \emph{{IEEE} Trans. Wireless Commun.}, vol.~10, no.~9, pp.
  3020--3031, Sep. 2011.

\bibitem{Zhao2007}
Q.~Zhao, L.~Tong, A.~Swami, and Y.~Chen, ``Decentralized cognitive {MAC} for
  opportunistic spectrum access in ad hoc networks: A {POMDP} framework,''
  \emph{{IEEE} J. Sel. Areas Commun.}, vol.~25, no.~3, pp. 589--600, Apr. 2007.

\bibitem{Choi2011}
K.~W. Choi and E.~Hossain, ``Opportunistic access to spectrum holes between
  packet bursts: A learning-based approach,'' \emph{{IEEE} Trans. Wireless
  Commun.}, vol.~10, no.~8, pp. 2497--2509, Aug. 2011.

\bibitem{Digham2007}
F.~F. Digham, M.-S. Alouini, and M.~K. Simon, ``On the energy detection of
  unknown signals over fading channels,'' \emph{{IEEE} Trans. Commun.},
  vol.~55, no.~1, pp. 21--24, Jan. 2007.

\bibitem{Atapattu2009}
S.~Atapattu, C.~Tellambura, and H.~Jiang, ``Energy detection of primary signals
  over $\eta$- $\mu$; fading channels,'' in \emph{Proc. International
  Conference on Industrial and Information Systems (ICIIS)}, Dec. 2009, pp.
  118--122.

\bibitem{Lai2008}
L.~Lai, Y.~Fan, and H.~V. Poor, ``Quickest detection in cognitive radio: {A}
  sequential change detection framework,'' in \emph{Proc. {IEEE} Global
  Telecommunications Conference ({Globecom})}, New Orleans, LO, USA, Dec. 2008.

\bibitem{Haghighi20112}
K.~Haghighi, E.~{Str{\"{o}}m}, and E.~Agrell, ``An {LLR-based} cognitive
  transmission strategy for higher spectrum reutilization,'' in \emph{Proc.
  {IEEE} Global Telecommunications Conference ({Globecom})}, Houston, Texas,
  USA, Dec. 2011.

\bibitem{Yu2003}
S.-Z. Yu and H.~Kobayashi, ``An efficient forward-backward algorithm for an
  explicit-duration hidden {Markov} model,'' \emph{{IEEE} Signal Process.
  Lett.}, vol.~10, no.~1, pp. 11--14, Jan. 2003.

\bibitem{neyman1933}
J.~Neyman and E.~S. Pearson, ``On the problem of the most efficient tests of
  statistical hypotheses,'' \emph{Philosophical Trans. of the Royal Soc. of
  London. Series A}, vol. 231, pp. 289--337, 1933.

\bibitem{kendall1973}
M.~Kendall, A.~Stuart, and J.~Ord, \emph{The advanced theory of statistics},
  ser. Griffin's statistical monographs and courses.\hskip 1em plus 0.5em minus
  0.4em\relax London: Griffin, UK, 1973, vol.~1.

\bibitem{henkbook2007}
H.~Wymeersch, \emph{Iterative receiver design}.\hskip 1em plus 0.5em minus
  0.4em\relax Cambridge, UK: Cambridge University Press, 2007.

\end{thebibliography}
\end{document}